\documentclass[letterpaper, 10 pt, conference]{ieeeconf}  

\IEEEoverridecommandlockouts                              
\usepackage{amsmath}
\usepackage{amssymb}
\usepackage{float}
\usepackage{graphicx}
\usepackage{subfigure}
\usepackage{epstopdf}
\usepackage{color}
\usepackage{verbatim}
\usepackage{tikz,times}
\usepackage{algorithm}
\usepackage{algorithmic}
\usepackage{cite}

\DeclareMathOperator{\diag}{diag}

\newcommand{\norm}[1]{\left\lVert#1\right\rVert}

\newtheorem{detector}{\textbf{Detector}}

\newtheorem{theorem}{\textbf{Theorem}}
\newtheorem{proposition}{\textbf{Proposition}}
\newtheorem{lemma}{\textbf{Lemma}}

\newtheorem{assumption}{\textbf{Assumption}}

\title{\LARGE \bf
	Secure Platooning  of Autonomous Vehicles \\ Under Attacked GPS Data
}

\author{Xingkang~He,  
	Ehsan Hashemi, 
	Karl H. Johansson 
	\thanks{The work is supported by Knut \& Alice Wallenberg foundation, and by Swedish Research Council.}
	\thanks{Xingkang He, and Karl H. Johansson are with Division of Decision and Control Systems, School of Electrical Engineering and Computer Science. KTH Royal Institute of Technology, Sweden ((xingkang,kallej)@kth.se).}%
	\thanks{Ehsan Hashemi is with the Department of Mechanical and Mechatronics Engineering, University of Waterloo, Waterloo, ON, Canada (ehashemi@uwaterloo.ca).}
}

\begin{document}

	\maketitle
	\thispagestyle{empty}
	\pagestyle{empty}

	\begin{abstract}
		In this paper, we study how to secure the platooning of autonomous vehicles when an unknown vehicle is under attack and bounded system uncertainties exist. For the attacked vehicle, its position and speed 
		measurements from GPS can be manipulated arbitrarily by a malicious attacker. 
		First, to find out which vehicle is under attack, two  detectors are proposed by using the relative measurements (by camera or radar) and the local innovation obtained through measurements from neighboring vehicles. 
		Then, based on the results of the detectors, we design a local state observer for each vehicle by applying a saturation method to the measurement innovation. 
		Moreover, based on the neighbor state estimates provided by the observer, a distributed controller is proposed to achieve the  consensus in vehicle speed and keep fixed desired distance between two neighboring vehicles. The estimation error by the observer and the platooning error by the controller are shown to be asymptotically upper bounded under certain conditions.
		The effectiveness of the proposed methods is also evaluated in numerical simulations.
	\end{abstract}

	\section{Introduction}
	The potential to enable fast and reconfigurable mechanisms for increasingly prevalent Automated Driving Systems (ADS), cooperative intelligent transportation systems (C-ITS), and vehicle platooning, as well as compelling mobility and safety benefits, underscores the critical need to develop more reliable and secure distributed-system designs in intelligent transport.
	In vehicle platooning and cooperative adaptive cruise controls (CACC), estimating the current vehicle states 
	, which are shared via communication between vehicles, is essential for stability and traction control, as well as motion planning \cite{milanes2014cooperative, liang2016heavy, siampis2017real, jalali2017integrated}. Vehicle (longitudinal and lateral) speeds can be measured by a GPS, but their reliability due to loss of reception and the poor accuracy of available commercial GPSs, specifically in measuring sideslip (and lateral speed), necessitates developing a reliable state estimators in connected ADS that share measurements and local estimates over Dedicated Short-Range Communications (DSRC) or 5G NR access via 3GPP PC5 or IEEE 802.11 technologies.
	
	Vehicle platooning has enhanced in recent years in terms of reliability through vehicle-to-vehicle (V2V) connectivity, distributed state estimation, and learning-aided controls \cite{cui2018development, ucar2018ieee, jin2019adaptive}. However, the existing secure platooning solutions in C-ITS are prohibitively inefficient in dealing with malicious attacks on GPS measurements. Such attacks could be through transmission of the received GPS data (in data acquisition modules) to ADS' control systems or the GPS receiver. Thus, reliable state estimation, resilient to attacks and robust to system uncertainties,
	plays key roles in ensuring the safety of intelligent transport by improving the reliability of vehicle local active safety systems and enhancing the performance in cooperative tasks such as CACC \cite{turri2017cooperative, van2017analyzing}.
	Real-time methods are proposed in \cite{lyamin2014real, lyamin2019real} to diagnose jamming Denial-of-Service (DoS) attacks in IEEE 802.11p vehicular networks and its false alarm probabilities are estimated.
	A data-driven fault detection approach and a decision support system are proposed in \cite{sargolzaei2016machine} to diagnose attacks and track fault data injection attacks in CACC (in real time).
	
	To mitigate the DoS attack in CACC, a set of linear Luenberger observers is designed in \cite{biron2017resilient} by using LMI over delayed measurements to develop a resilient control strategy.
	An adaptive control strategy is designed in \cite{jin2018adaptive} to deal with time-invariant sensor and actuator attacks in a vehicular network, with wireless V2V communication, while guaranteeing uniform boundedness of the closed-loop system.
	A decentralized proportional-derivative controller augmented by a triggering mechanism for getting preceding vehicles' new measurements is designed in \cite{merco2019resilient} to maintain string stability in a vehicular platoon.
	A trust-based service recommendation scheme is proposed in \cite{hu2016replace} to avoid selecting badly behaved head vehicles for  ballot-stuffing and on–off attacks in vehicular ad hoc networks.
	In \cite{petrillo2018collaborative, pirani2018resilient} cooperative control protocols and state observers are provided for enhancing resilience to attacks and detecting faults in vehicular platoons.


	To this end, by using available secure radar and stereo camera, which are available in autonomous vehicles, ADS, and advanced driver-assistance systems, and the potentially attacked GPS data in an unknown vehicle,
	this paper studies how to design a secure algorithm such that a group of autonomous vehicles achieve practical platooning.
	The main contributions of this paper are summarized in the following.
	\begin{enumerate}
		\item To find out which vehicle is under attack, we propose two detectors by using the relative measurements and the local innovation obtained through the absolute  and  relative measurements of three neighboring vehicles, respectively. Then, based on the results of the detectors, we design a local state observer for each vehicle by applying a saturation method to the measurement innovation. 
		
		
		
		\item Based on the neighbor state estimates provided by the observer, we design a distributed controller to achieve consensus in vehicle speed and keeping fixed desired distance between two neighboring vehicles. It is also shown that the estimation error (by the observer) and the platooning error are asymptotically upper bounded  under certain conditions.
		
	\end{enumerate}
	
	
	
	
	The remainder of the paper is organized as follows: Section  \ref{sec_formulation} is on the problem formulation. Section \ref{sec_algorithm} provides the secure platooning algorithm, whose  performance is studied in   Section \ref{sec:analysis}. 
	After  numerical simulations in Section \ref{sec:simulation}, the paper is concluded in   Section \ref{sec_conclusion}.
	The main proofs are given in Appendix.

	\subsection{Notations}
	The superscript ``T" represents the transpose. $\mathbb{R}^{n\times m}$ is the set of real matrices with $n$ rows and $m$ columns.  $\mathbb{R}^n$ is the $n$-dimensional Euclidean space.
	$I_{n}$ stands for the $n$-dimensional square identity matrix. 
	$\diag\{\cdot\}$   represents the diagonalization operator.  
	$A\otimes B$ is the Kronecker product of $A$ and $B$.  $\norm{x}$ is the 2-norm of a vector $x$. $\norm{A}$ is the induced 2-norm, i.e., $\norm{A}=\sup\limits_{x\neq 0}\frac{\norm{Ax}}{\norm{x}}$.  
	$\lambda_{\min}(A)$ and $\lambda_{\max}(A)$ are the  minimal and maximal eigenvalues of a real-valued symmetric matrix  $A$, respectively.  $I_{i\in \mathcal{C}}$ stands for the indicator function, which equals 1 if $i\in \mathcal{C}$, and 0 otherwise.

	\section{Problem Formulation}\label{sec_formulation}
	\subsection{System model}
	Considering $N\geq 3$ vehicles, without loss of generality, we assume that the order of these vehicles from the leader to the tail is $1,2,\dots,N$.  
	The leader vehicle, i.e., vehicle $1,$  satisfies the following dynamics:
	\begin{align}\label{eq_local_state}
	x_1(t+1)=Ax_1(t)+n_{1}(t),
	\end{align}
	and  vehicle $i$, $i=2,3,\dots,N$, satisfies
	\begin{align}\label{eq_local_state2}
	x_i(t+1)=Ax_i(t)+[0,Tu_i(t)]^T+n_{i}(t),
	\end{align}
	where  $x_j(t)=[s_j(t),v_j(t)]^T\in\mathbb{R}^2$ is the state of vehicle $j$ consisting of position $s_j(t)$ and velocity $v_j(t)$, 
	$n_{j}(t)\in\mathbb{R}^2$  the process noise,  $u_i(t)\in\mathbb{R}$  the control input,  $A=\left(\begin{smallmatrix}
	1&T\\
	0&1
	\end{smallmatrix}\right)$,
	and $T$  the sampling time, $j=1,2,\dots,N$.

	Vehicle $j\in\{1,2,\dots,N\}$,  obtains its position and velocity measurements from GPS:
	\begin{equation}\label{eq_system2}
	\begin{split}
	y_{j,j}(t)&=x_j(t)+a_{j}(t)+d_{j,j}(t)
	\end{split}
	\end{equation}
	where
	$y_{j,j}(t)\in\mathbb{R}^2$ and $d_{j,j}(t)\in\mathbb{R}^2$ are the measurement and measurement noise of vehicle $j$. $a_{j}(t)\in\mathbb{R}^2$ is the  attack signal injected by a malicious attacker.

	Moreover, by employing the vehicle sensors (e.g., radar and camera), vehicle $i\in\{2,3,\dots,N\}$ is able to measure the relative state between itself and its front vehicle (i.e., vehicle $i-1$):
	\begin{equation}\label{eq_system3}
	\begin{split}
	y_{i-1,i}(t)&=x_i(t)-x_{i-1}(t)+d_{i-1,i}(t), 
	\end{split}
	\end{equation}
	where $y_{i-1,i}(t)\in\mathbb{R}^2$ and $d_{i-1,i}(t)\in\mathbb{R}^2$ are the measurement and measurement noise. Note that  in this paper we assume the relative state measurements are attack-free.

	\begin{assumption}\label{ass_noise}
		On the noise processes in \eqref{eq_local_state}--\eqref{eq_system3} and the initial estimate $\hat x_{j}(0)$, for $j\in\{1,\dots,N\}, $ and $ i\in\{2,\dots,N\},$ it holds that
		\begin{align*}
		\norm{\hat x_{j}(0)-x_j(0)}&\leq q,\\
		\norm{n_j(t)}&\leq \epsilon, 	\\
		\max\{\norm{d_{j,j}(t)},\norm{d_{i-1,i}(t)}\}&\leq \mu, 	 
		\end{align*}
		where $q,\epsilon$, and $\mu$ are positive real-valued scalars known to each vehicle.
	\end{assumption}
	%

	\subsection{Attack model}
	Denote $\mathcal{V}=\{1,2,\dots,N\}$, $\mathcal{A}$ the label set of the attacked vehicle, and $\mathcal{A}^c=\mathcal{V}-\mathcal{A}$  the label set of attack-free vehicles.

	\begin{assumption}\label{ass_attack}
		There is an unknown vehicle $j\in\mathcal{V}$ under attack, i.e., $\mathcal{A}=\{j\},$ such that 
		the attack signal $a_{i}(t)$ in \eqref{eq_system2} satisfies 
		\begin{align*}
		a_{i}(t)&\in\mathbb{R}^2, i=j,\\
		a_{i}(t)&=0, i\neq j.
		\end{align*}
	\end{assumption}
	
	From Assumption \ref{ass_attack}, the GPS measurements of the attacked vehicle $j$ can be manipulated arbitrarily. There is tradeoff on the use of the measurements. On one hand, if the measurements are fully trusted by conventional observers, once the vehicle is under attack, the estimation performance would be seriously degraded by the injected attack signals. On the other hand, if little confidence is given to the measurements, the estimation performance would be also degraded, since little information is utilized. Therefore, it is worth studying how to employ the potentially attacked measurements and how to design detectors to find out the attacked vehicle.

	\subsection{Communication of vehicles}
	We denote $\bar{\mathcal{N}}_i$ the neighbor set of vehicle $i$, where $\bar{\mathcal{N}}_i=\{i-1,i\}$ if $i\in\{2,\dots,N-1\}$, $\bar{\mathcal{N}}_1=\{2,3\}$, and  $\bar{\mathcal{N}}_N=\{N-2,N-1\}$. The communication topology of the vehicles is illustrated in Fig.~\ref{fig:Model_Description} with $N=7$, where each vehicle $i$, $i\in\{1,2,\dots,N\}$, can receive the measurements of its neighboring vehicle $j$, $j\in\bar{\mathcal{N}}_i$, including both absolute and relative measurements in \eqref{eq_system2} and \eqref{eq_system3}.
	Note that for the leader and tail vehicles, i.e., vehicle $1$ and  vehicle $N$, to ensure the redundancy of measurement information against the possible attacks on them, we require  that vehicle $1$ can  obtain the measurements from  vehicles $2$ and $3$, and that vehicle $N$   can  obtain the measurements from vehicles $N-1$ and $N-2$.

	
	
	\begin{figure}[t]
		\centering
		\includegraphics[width=0.99\linewidth]{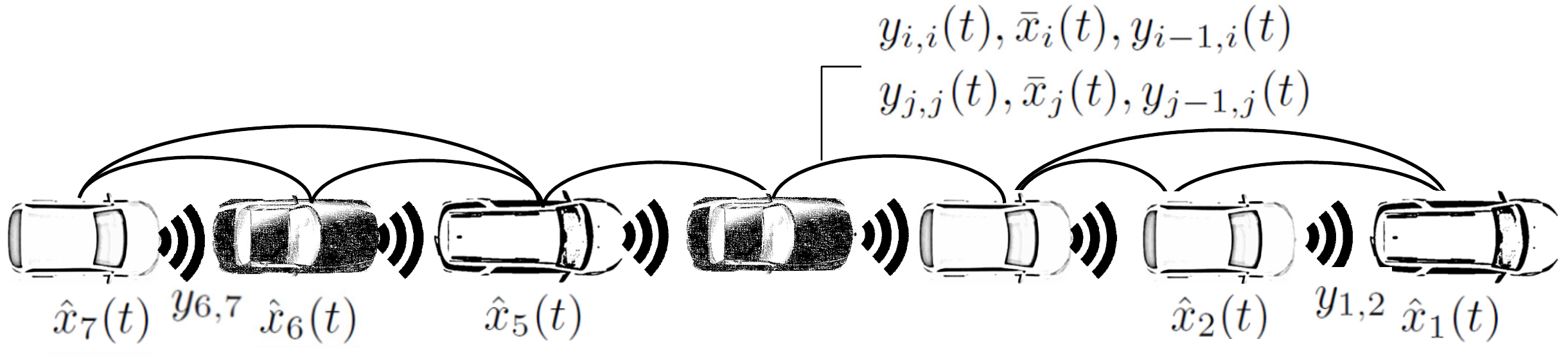}
		\caption{Platoon model and communication topology of seven vehicles: Each vehicle is able to measure the relative state between itself and its front vehicle by radar or camera (i.e.,  $y_{j,i}$, $j\in\bar{\mathcal{N}_i}$). For two  vehicles connected  by wireless communications (e.g., WIFI), they can exchange their absolute and relative measurements (i.e., $y_{i,i}$ and $y_{j,i}$, $j\in\bar{\mathcal{N}_i}$), and the predicted estimates of own states (i.e., $\bar x_i(t)$).}
		\label{fig:Model_Description}
	\end{figure}

	\textbf{Problem:} We aim to design the control input $u_i(t)$ for vehicles $i,2\leq i\leq N,$ such that their speeds are close to the speed of the leader vehicle $1$ and two neighboring vehicles keep a certain distance in position, i.e.,  
	\begin{align}\label{eq_platoon}
	\begin{split}
	&\limsup\limits_{t\rightarrow \infty}|v_i(t)-v_1(t))|\leq \gamma_1,\\
	&\limsup\limits_{t\rightarrow \infty}|s_i(t)- s_{i-1}(t)-\Delta_{i,i-1}|\leq \gamma_2,
	\end{split}
	\end{align}
	against the attack signal $\{a_i(t),i\in\mathcal{A}\}$ injected by some malicious attacker,
	where $\gamma_1$ and $\gamma_2$ are nonnegative real-valued scalars, related to the system noise.
	\section{Co-design of observer,  detector, and controller}\label{sec_algorithm}
	In this section, for each vehicle $i$,  $i\in\{2,\dots,N\}$,  we will first design a local state observer to estimate its velocity and position, then develop two local attack detectors to find out whether it is under attack, and finally provide a distributed controller by employing the estimates from neighboring vehicles to achieve the vehicle platooning in \eqref{eq_platoon}.

	\subsection{Observer design}

	%
	
	
	

	
	
	By \eqref{eq_system4}--\eqref{eq_system_1},  the measurement equation for  vehicle $i$,  $i\in\{1,\dots,N\}$,  can be written in the following:
	\begin{equation}\label{eq_system5}
	\begin{split}
	z_i(t)&=Cx_i(t)+\boldsymbol{a}_i(t)+\boldsymbol{d}_i(t)
	\end{split}
	\end{equation}
	where 
	\begin{align}\label{notation}
	\begin{split}
	z_1(t)&=[y_{1,1}^T(t), y_{1|2}^T(t),y_{1|3}^T(t)]^T,\\
	\boldsymbol{a}_1(t)&=[a_{1}^T(t), a_{2}^T(t), a_{3}^T(t)]^T,\\
	\boldsymbol{d}_1(t)&=[d_{1,1}^T(t), d_{1|2}^T(t),d_{1|3}^T(t)]^T,\\
	z_i(t)&=[y_{i|i-1}^T(t), y_{i,i}^T(t),y_{i|i+1}^T(t)]^T,2\leq i\leq N-1,\\
	\boldsymbol{a}_i(t)&=[a_{i-1}^T(t), a_{i}^T(t), a_{i+1}^T(t)]^T,\\
	\boldsymbol{d}_i(t)&=[d_{i|i-1}^T(t), d_{i,i}^T(t),d_{i|i+1}^T(t)]^T,\\
	z_N(t)&=[y_{N|N-2}^T(t), y_{N|N-1}^T(t),y_{N,N}^T(t)]^T,\\
	\boldsymbol{a}_N(t)&=[a_{N-2}^T(t), a_{N-1}^T(t),a_{N}^T(t)]^T,\\
	\boldsymbol{d}_N(t)&=[d_{N|N-2}^T(t), d_{N|N-1}^T(t),d_{N,N}^T(t)]^T,\\
	C&=\begin{pmatrix}
	I_{2\times 2}\\
	I_{2\times 2}\\
	I_{2\times 2}
	\end{pmatrix},
	\end{split} 
	\end{align}
	where the notations in \eqref{notation} are given in  Appendix \ref{app_derivation}. Note that $y_{i|j}$, $j=i-1,i+1$ stands for the absolute measurement of vehicle $i$ calculated with the relative measurement between vehicles $i$ and $j$.  
	
	Since the leader vehicle is control-free, i.e., $u_1(t)=0$, from \eqref{eq_local_state}, \eqref{eq_local_state2}, and  \eqref{eq_system5},  we have the reformulated state equation  and  measurement equation  of vehicle $i,i\in\{1,\dots,N\},$ in the following
	\begin{align}\label{system}
	\begin{split}
	x_i(t+1)&=Ax_i(t)+[0,Tu_i(t)]^T+n_{i}(t)\\
	z_i(t)&=Cx_i(t)+\boldsymbol{a}_i(t)+\boldsymbol{d}_i(t).
	\end{split}
	\end{align}
	We aim to design an observer for vehicle $i$ with two steps, namely, time update and measurement update.
	In the time update, for vehicle $i$, we let
	\begin{align}\label{eq_pred}
	\bar x_i(t)=A\hat x_i(t-1)+[0,Tu_i(t-1)]^T,
	\end{align}
	where  $\hat x_i(t)$ is the estimate of $x_i(t)$.
	Then we denote the measurement innovation of vehicle $i$  at time $t$ by $\eta_i(t)$, where
	\begin{align}\label{eq_innovation}
	\eta_i(t)=z_i(t)-C\bar x_i(t).
	\end{align}
	
	To number the labels of vehicles in the measurement update equation, we let
	\begin{align}\label{eq_notation}
	\eta_i(t)&=:\left(\begin{smallmatrix}
	\eta_{i,j_1}^{[1]}(t)&\eta_{i,j_1}^{[2]}(t)&\eta_{i,j_2}^{[1]}(t)&\eta_{i,j_2}^{[2]}(t)&\eta_{i,j_3}^{[1]}(t)&\eta_{i,j_3}^{[2]}(t)
	\end{smallmatrix}\right)\nonumber\\
	K_i(t)&=:\diag\{k_{i,j_1}^{[1]}(t),k_{i,j_1}^{[2]}(t),k_{i,j_2}^{[1]}(t),k_{i,j_2}^{[2]}(t),\\
	&\qquad\qquad\qquad k_{i,j_3}^{[1]}(t),k_{i,j_3}^{[2]}(t)\},\nonumber
	\end{align}
	where $j_1,j_2,j_3\in\bar{\mathcal{N}}_i\cup\{i\}$. For example, from equation \eqref{eq_system5}, if $i=1$, $j_1=1,j_2=2,j_3=3$, and if $i=N$, $j_1=N-2,j_2=N-1,j_3=N$, otherwise,  $j_1=i-1,j_2=i,j_3=i+1$.
	Then  we design a saturation-like scheme to utilize the measurement $\eta_i(t)$ as follows.
	
	For each $m=\{1,2,3\}$, and for each $r=\{1,2\}$, 
	\begin{align}\label{eq_K}
	k_{i,j_m}^{[r]}(t)=\begin{cases}
	1,\qquad\qquad\text{ if } |\eta_{i,j_m}^{[r]}(t)|\leq\beta,\\
	\frac{\beta}{|\eta_{i,j_m}^{[r]}(t)|}, \quad\text{ otherwise},
	\end{cases}
	\end{align}
	where $\beta\geq 0$ is a saturation parameter to be designed in the following. 
	From equations \eqref{eq_notation} and \eqref{eq_K},  the update equation is given:
	\begin{align}\label{alg_update}
	\hat x_{i}(t)=
	\bar x_{i}(t)+ \frac{1}{2}C^TK_i(t)\eta_i(t), i\in\{1,\dots,N\}.
	\end{align}

	\subsection{Detector design}
	Although the saturation based estimation method provided in the previous subsection can be employed directly to estimate vehicle state (position and speed), under potentially attacked GPS measurements, an attack detection protocol is provided in the following to identify the vehicle under attack. When the attacked vehicle is identified, we provide an attack-free estimation  approach by isolating the measurements of the attacked vehicle to improve the estimation performance. 
	
	The  idea of the first detector is that for two attack-free vehicles, the norm  of the difference between the relative measurements and the absolute measurements is upper bounded by a scalar related to the noise bound. Otherwise, one of the two vehicles is under attack.
	
	\begin{lemma}\label{lem_relative}
		For each attack-free vehicle $i\geq 2$ and $i-1$, i.e., $i\in\{2,\dots,N\}\cap\mathcal{A}^c$,  if Assumption \ref{ass_noise} holds,
		\begin{align*}
		\norm{f_{i,i-1}(t)}\leq 3\mu,
		\end{align*}
		where $f_{i,i-1}(t)=y_{i-1,i}(t)+y_{i-1}(t)-y_i(t)$.
	\end{lemma}
	\begin{proof}
		By equation \eqref{eq_system2}, for $i\in\{2,\dots,N\}\cap\mathcal{A}^c$,
		\begin{equation}\label{pf_lem1}
		\begin{split}
		y_{i,i}(t)&=x_i(t)+d_{i,i}(t)\\
		y_{i-1,i-1}(t)&=x_{i-1}(t)+d_{i-1,i-1}(t).
		\end{split}
		\end{equation}
		From equation \eqref{eq_system3},
		\begin{equation}\label{pf_lem2}
		\begin{split}
		y_{i-1,i}(t)&=x_i(t)-x_{i-1}(t)+d_{i-1,i}(t).
		\end{split}
		\end{equation}
		Then substituting \eqref{pf_lem1} into \eqref{pf_lem2} yields
		\begin{equation}\label{pf_lem3}
		\begin{split}
		f_{i,i-1}(t)&=d_{i-1,i}(t)+d_{i-1,i-1}(t)+d_{i,i}(t).
		\end{split}
		\end{equation}
		Taking the norm of both sides of \eqref{pf_lem3} leads to the conclusion.
	\end{proof}
	
	Let $\Gamma_i(t)$ be the detected  set of the attacked vehicle  by vehicle $i$ at time $t$, which contains the label of the attacked vehicle.
	Let $\Theta_i(t)$ be the detected-with-doubt vehicle set at time $t$  by vehicle $i$,  which contains the labels of two vehicles and one of them is the attacked vehicle. At the initial time, we assume   $\Gamma_i(0)=\emptyset,$ and  $\Theta_i(0)=\emptyset,$ $i\in\{1,\dots,N\}$. The  sets $\Gamma_i(t)$ and $\Theta_i(t)$ are shared between neighboring vehicles. From Lemma \ref{lem_relative}, we propose the following detector to identify the attacked vehicle.
	\begin{detector}\label{detector1}
		For vehicle $i\in\{2,\dots,N-1\}$, test the inequalities in \ref{item1}) and \ref{item2}); for vehicle $i=1$, test the inequality in  \ref{item2}); for vehicle $i=N$, test the inequality in \ref{item1}): 
		\begin{enumerate}
			\item $\norm{f_{i,i-1}(t)}>3\mu$\label{item1}
			\item $\norm{f_{i+1,i}(t)}>3\mu$\label{item2}
		\end{enumerate}
		then 
		\begin{itemize}
			\item 
			if there is a time $t$ such that \ref{item1}) holds, either vehicle $i$ or $i-1$ is under attack, and we let $\Theta_i(t)=\{i-1,i\}$;
			\item 
			if there is a time $t$ such that \ref{item2}) holds, either vehicle $i$ or $i+1$ is under attack, and we let $\Theta_i(t)=\{i,i+1\}$;
			\item if there are two times under which \ref{item1}) and \ref{item2}) hold respectively, then vehicle $i$ is claimed to be under attack and we let $\Gamma_i(t)=\{i\}$;
		\end{itemize}
	\end{detector} 
	\begin{proposition}
		Under Assumptions \ref{ass_noise}--\ref{ass_attack}, for Detector \ref{detector1}, all the three claims hold.
	\end{proposition}
	\begin{proof}
		From Assumption \ref{ass_noise} and Lemma \ref{lem_relative}, the first and second claims hold. Furthermore, under Assumption \ref{ass_attack}, the third  claim holds. 
	\end{proof}

	The idea of the second detector is that for an attack-free vehicle, its measurement innovation is upper bounded by a sequence related to the estimation error. Otherwise, this vehicle is under attack. 
	Define the sequence $\{\rho(t)\}$ in the following
	\begin{align}\label{sequ}
	\rho(t+1)=(1-k(t+1))\norm{A}\rho(t)+Q,
	\end{align}
	where $\rho(0)=q,$ and 
	\begin{align}\label{k_seq}
	\begin{split}
	k(t+1)=&\min\{1,\frac{\beta}{\norm{A}\rho(t)+\epsilon+\mu}\},\\
	Q=&\frac{3}{2}(\epsilon+\mu)+\frac{\sqrt{2}}{2}\beta.
	\end{split}
	\end{align}
	\begin{lemma}\label{lem_detec2}
		Under Assumption \ref{ass_noise}, for each  vehicle $i\in\{1,\dots,N\}$, 
		\begin{itemize}
			\item $\norm{\hat x_{i}(t)-x(t)}\leq \rho(t)$;
			\item if vehicle $i$ is attack-free, i.e., $i\in\mathcal{A}^c$, then for $t\geq 1,$
			$$	\norm{y_{i,i}(t)-A\hat x_{i}(t-1)}\leq \norm{A}\rho(t-1)+\epsilon+\mu.$$
		\end{itemize}
		where $\rho(t)$ is in \eqref{sequ}.
	\end{lemma}
	\begin{proof}
		See  Appendix.
	\end{proof}
	Based on Lemma \ref{lem_detec2}, we provide the following detector.
	\begin{detector}\label{detector2}
		We claim vehicle $i$ is  under attack and let $\Gamma_i(t)=\{i\}$, if 
		\begin{align*}
		\norm{y_{i,i}(t)-A\hat x_{i}(t-1)}>\norm{A}\rho(t-1)+\epsilon+\mu,
		\end{align*}
		where $\rho(t-1)$ is in \eqref{sequ}.
	\end{detector}
	%

	If the attacked vehicle, e.g., vehicle $s\in\Gamma_i$, is found out by Detector \ref{detector2}, or two potentially attacked vehicles are found out by Detector \ref{detector1}, e.g., vehicles $l,l+1\in\Theta_i(t)$,   to remove the influence of the measurements of vehicle $s$ or vehicles $l,l+1$ to the estimation performance,   the observer gains ${K_{i}(t)}_{i=1}^N$ of all the vehicles are supposed to be adjusted as follows. Recall the form of $K_i(t)$ in \eqref{eq_notation}: 
	\begin{align}\label{eq_notation2}
	K_i(t)&=:\diag\{k_{i,j_1}^{[1]}(t),k_{i,j_1}^{[2]}(t),k_{i,j_2}^{[1]}(t),k_{i,j_2}^{[2]}(t),\\
	&\qquad\qquad\qquad k_{i,j_3}^{[1]}(t),k_{i,j_3}^{[2]}(t)\}.\nonumber
	\end{align}
	Then we provide the design of the elements $k_{i,j_m}^{[r]}(t)$, $r=1,2,$ and $m=1,2,3$ in the following.
	
	If $\Gamma_i(t)=\{s\}$, let 
	\begin{align}\label{eq_K20}
	\begin{split}
	k_{i,s}^{[1]}(t)&=0,k_{i,s}^{[2]}(t)=0,\forall s\in\hat{\mathcal{N}}_i,\\
	k_{i,l}^{[1]}(t)&=1,k_{i,l}^{[2]}(t)=1, l\neq s, l\in \hat{\mathcal{N}}_i,
	\end{split}
	\end{align}
	where $\hat{\mathcal{N}}_i=\bar{\mathcal{N}}_i\cup \{i\}.$
	
	If $\Gamma_i(t)= \emptyset$ but $\Theta_i(t)\neq \emptyset$, let
	\begin{align}\label{eq_K21}
	\begin{split}
	k_{i,s}^{[1]}(t)&=0,k_{i,s}^{[2]}(t)=0,\forall s\in\hat{\mathcal{N}}_i\cap \Theta_i(t),\\
	k_{i,l}^{[1]}(t)&=1,k_{i,l}^{[2]}(t)=1, l\neq s, l\in \hat{\mathcal{N}}_i.
	\end{split}
	\end{align}
	In other cases, the design of $K_{i}(t)$ remains the same as equation \eqref{eq_K}, i.e.,
	for each $m=\{1,2,3\}$, and for each $r=\{1,2\}$, 
	\begin{align}\label{eq_K2}
	k_{i,j_m}^{[r]}(t)=\begin{cases}
	1,\qquad\qquad\text{ if } |\eta_{i,j_m}^{[r]}(t)|\leq\beta,\\
	\frac{\beta}{|\eta_{i,j_m}^{[r]}(t)|}, \quad\text{ otherwise}.
	\end{cases}
	\end{align}
	Based on the designed $K_i(t)$, all the measurements of the attack-free vehicles will not be saturated and the measurements of the possible attacked vehicles (e.g., the detected vehicles) will no longer be employed.


	\subsection{Controller design}
	Next, we aim to design the control input based on the neighbor estimates and the desired relative position distance  between two neighboring vehicles. 
	Denote $\mathcal{N}_i$ the set of vehicle(s) nearest  to vehicle $i$, i.e., 
	\begin{align}\label{eq_measurement_notation}
	\mathcal{N}_i=
	\begin{cases}
	\{i-1,i+1\},&\text{if } i\in\{2,3,\dots,N-1\}\\
	\{2\}, &\text{if } i=1\\
	\{N-1\}, &\text{if } i=N.
	\end{cases}
	\end{align}
	Note that $\mathcal{N}_i\subset \bar{\mathcal{N}}_i$, $\forall i\in\{1,2,\dots,N\}.$

	Assume  $\hat s_{i}(t)$ and $\hat v_{i}(t)$ are the estimates of  $s_{i}(t)$ and $v_{i}(t)$, respectively. Let $t_i^*$ be the start time for the control input. For $0\leq t< t_i^*$, we let $u_{i}(t)=0.$ 
	Regarding the acceleration $u_i(t)$ in  \eqref{eq_pred}, by employing the predicted estimates $\bar x_j(t)=[\bar s_{j}(t),\bar v_{j}(t)]^T$ from the vehicle $j\in\mathcal{N}_i$, vehicle $i$ is equipped with the following acceleration input, for $t\geq t_i^*\geq 0$,  
	\begin{align}\label{eq_control}
	\begin{split}
	u_{i}(t)=&\sum_{j\in\mathcal{N}_i}\big(g_s(\bar s_{j}(t)-\hat s_{i}(t)+\Delta_{i,j})\\
	&+g_v(\bar v_{j}(t)-\hat v_{i}(t))\big), 2\leq i\leq N,
	\end{split}
	\end{align}
	where $g_s$ and $g_v$ are positive scalars to be determined.
	Note that if the initial estimates are very accurate, we can choose $t_i^*=0$. Otherwise, setting a larger $t^*$ can lead to better estimates but need more time to achieve the platooning of vehicles.
	

	Based on the observer, the two detectors and the distributed controller, for the system \eqref{system}, we propose the secure platooning algorithm for vehicle $i$ in Algorithm \ref{alg:A}.

	%
	%

	\begin{algorithm}
		\caption{Secure Autonomous Vehicle Platooning:}
		\label{alg:A}
		\begin{algorithmic}[1]
			\STATE {\textbf{Initial setting:} ($\bar x_{i}(0),\beta,g_s,g_v$)}\\		\vskip 2pt
			\STATE {\textbf{Communications between neighboring vehicles:} } Vehicle $i$ sends out its measurements $y_{i-1,i}(t)$, $y_{i,i}(t)$, the sets $\Theta_i(t-1)$ and $\Gamma_i(t-1)$, and the estimate $\bar x_i(t)$ to its neighboring vehicle $j\in\bar{\mathcal{N}}_i$. 
			Let $\Gamma_i(t)=\cup_{j\in\bar{\mathcal{N}}_i}\Gamma_j(t-1)\cup \Gamma_i(t-1)$ and $\Theta_i(t)=\cup_{j\in\bar{\mathcal{N}}_i}\Theta_j(t-1)\cup \Theta_i(t-1)$.
			Then vehicle $i$ runs Detector \ref{detector1}. \\ \vskip 2pt
			\STATE {\textbf{Measurement update:} Vehicle $i$, $i\in\{1,\dots,N\}$, uses the measurement $z_i(t)$ in \eqref{eq_system5} to update its own state estimate}
			$$\eta_i(t)=z_i(t)-C\bar x_i(t).$$
			Run Detector \ref{detector2} to update $\Theta_i(t)$ and $\Gamma_i(t)$. \\
			The estimate is updated in the following,
			\begin{align*}
			\hat x_{i}(t)&=
			\bar x_{i}(t)+ \frac{1}{2} C^TK_i(t)\eta_i(t),
			\end{align*}		
			where $K_i(t)$ is designed by  \eqref{eq_notation2}--\eqref{eq_K2}.
			\STATE {\textbf{Control input:} Vehicle $i$ uses its own estimate $\hat x_i(t)$ and the predicted estimates $\bar x_j(t),j\in\mathcal{N}_i,$ from its neighbors to design the input $u_i(t)$: $u_{i}(t)=0$, if $0\leq t\leq t_i^*$;  for $t\geq t_i^*$, }
			\begin{align*}
			\begin{split}
			u_{i}(t)=&\sum_{j\in\mathcal{N}_i}\big(g_s(\bar s_{j}(t)-\hat s_{i}(t)+\Delta_{i,j})\\
			&+g_v(\bar v_{j}(t)-\hat v_{i}(t))\big), 2\leq i\leq N.
			\end{split}
			\end{align*}
			\STATE {\textbf{Time update:} }\\	\vskip 2pt
			$\bar x_i(t+1)=A\hat x_i(t)+[0,Tu_i(t)]^T$.
		\end{algorithmic}
	\end{algorithm}
	
	\section{Algorithm Performance}\label{sec:analysis}
	In this section, we would like to study the performance of Algorithm \ref{alg:A}. First, we study the condition to ensure the asymptotically  bounded estimation error. Then, we provide the design for the control parameters such that the platooning error of vehicles are asymptotically bounded.
	
	In the following theorem, we study the boundedness of the estimation error for Algorithm \ref{alg:A}.
	\begin{theorem}\label{thm_estimation}
		Consider the system \eqref{system} satisfying Assumptions \ref{ass_noise}--\ref{ass_attack}, and Algorithm \ref{alg:A}. 
		If the following condition holds
		\begin{align}\label{condition0}
		(1-k^*)\norm{A}q+\frac{3}{2}(\epsilon+\mu)+\frac{\sqrt{2}}{2}\beta<q,
		\end{align}
		then   the estimation error is uniformly upper bounded, i.e., 
		\begin{align*}
		&\limsup\limits_{t\rightarrow \infty}\norm{\hat x_i(t)-x_i(t)}\leq \alpha, i\in\{1,\dots,N\},
		\end{align*}
		where $	k^*=\min\{1,\frac{\beta}{\norm{A}q+\epsilon+\mu}\}$, and $\alpha$ is given in  \eqref{alpha_2}.
	\end{theorem}
	\begin{proof}
		See Appendix.
	\end{proof}
	
	We note that if $\epsilon,\mu,T$ are sufficiently small, \eqref{condition0} is satisfied by choosing $\beta=p.$
	The boundedness of the vehicle platooning error is studied in the following theorem.
	\begin{theorem}\label{thm_control}
		Under the same conditions as Theorem \ref{thm_estimation}, 
		if 
		\begin{align}\label{eq_condition}
		\lambda_{\max}(\mathcal{L}_g)<\frac{4g_s}{g_v^2},	T<\frac{g_v}{g_s}
		\end{align}
		then for $i=2,\dots,N$, there exist constants $ \gamma_1$ and $\gamma_2$ such that
		\begin{align*}
		&\limsup\limits_{t\rightarrow \infty}|v_i(t)-v_1(t))|\leq \gamma_1,\\
		&\limsup\limits_{t\rightarrow \infty}|s_i(t)- s_{i-1}(t)+\Delta_{i,i-1}|\leq \gamma_2.
		\end{align*}
	\end{theorem}
	\begin{proof}
		See Appendix.
	\end{proof}

	%
	
	\section{Numerical Simulations}\label{sec:simulation}
	In this section, we study numerical simulations to show the effectiveness of the proposed algorithm in the vehicle platooning under attacked GPS data.
	
	Suppose there are five vehicles, in which the GPS measurement data of vehicle $3$ is compromised by a malicious attacker.  The attacker will insert the signal  $a_3(t)=2\tilde y_{3,3}(t))$, where $\tilde y_{3,3}(t)$ is the attack-free GPS measurement of  vehicle $3$.
	We suppose the time interval $t=[0,1000]$ with sampling step $T=1$ second.
	Both the process noise  $n_i(t)$ and measurement noise $d_{i,j}(t)$, $j\in\mathcal{N}_i\cup \{i\}$, $i=1,\dots,5$,   follow the uniform distribution between $[-0.1,0.1]$. The bounds in Assumption \ref{ass_noise} are assumed to be $\mu=0.1, \epsilon=0.1,q=100.5.$ The initial state is $x_1(0)=(100, 10)^T,$ $x_2(0)=(60, 8)^T,$ $x_3(0)=(40, 6)^T,$ $x_4(0)=(20, 4)^T,$ $x_5(0)=(0, 2)^T$, whose  estimates are all $0^{2\times 1}$. The required distance between vehicles $i$ and $j\in\mathcal{N}_i$ is $\Delta_{i,j}=20$.

	We conduct a Monto Carlo experiment with $100$ runs. 
	Define the average estimation error in position and velocity by $\eta_{i,s}(t)$ and $\eta_{i,v}(t)$, respectively, where
	\begin{align*}
	\eta_{i,s}(t)&=\frac{1}{100}\sum_{j=1}^{100}|e_{i,s}^j(t)|,\\
	\eta_{i,v}(t)&=\frac{1}{100}\sum_{j=1}^{100}|e_{i,v}^j(t)|,
	\end{align*}
	and $e_{i,s}^j(t)$ and $e_{i,v}^j(t)$ are the state estimation errors of vehicle $i$ in position and velocity, respectively, at time $t$ in the $j$th run.
	
	Moreover, we define the relative position and velocity between vehicle $i\in\{2,3,4,5\}$ and the leader vehicle $1$ by $\zeta_{i,s}(t)$ and $\zeta_{i,v}(t)$, respectively, where
	\begin{align*}
	\zeta_{i,s}(t)&=\frac{1}{100}\sum_{j=1}^{100}(s_{i}^j(t)-s_{1}^j(t)),\\
	\zeta_{i,v}(t)&=\frac{1}{100}\sum_{j=1}^{100}(v_{i}^j(t)-v_{1}^j(t)),
	\end{align*}
	and $s_{i}^j(t)$ and $v_{i}^j(t)$ are the position and velocity of vehicle $i$, respectively, at time $t$ in the $j$th run.
	
	First, we use a conventional algorithm, which has the same controller as Algorithm \ref{alg:A} but all $k_{i,j_m)}^{[r]}(t)=1,r=1,2,m=1,2,3,i\in\{1,\dots,5\}$ in the observer. For the case that the GPS data of vehicle $3$ is under attack, the simulation result is given in Fig. \ref{fig:convent_estimation}. It shows that both the estimation error and the platooning error show serious fluctuations. The average crash number\footnote{The crash number is the times where the order  from the leader vehicle $1$ to the tail vehicle $N$ is different from $1,2,\dots,N$. } of the vehicles is 179. 
	For Algorithm \ref{alg:A}, we choose $g_s=g_v=0.5$, which satisfies the condition in \eqref{eq_condition}. The control input starts at time $t_i^*=50.$ Also, we choose $\beta=1.$ The simulation results are provided in Fig. \ref{fig:estimation} and Fig. \ref{fig:platoon}. The attacked vehicle $3$ is successfully detected at $t=2$. Fig. \ref{fig:estimation} shows that the estimation errors of all vehicles tend to zero. In Fig. \ref{fig:platoon}, the relative positions of vehicles are ensured and all the velocities of the following vehicles tend to the velocity of the leader vehicle $1$.

	\begin{figure}[t]
		\centering
		\includegraphics[width=1\linewidth]{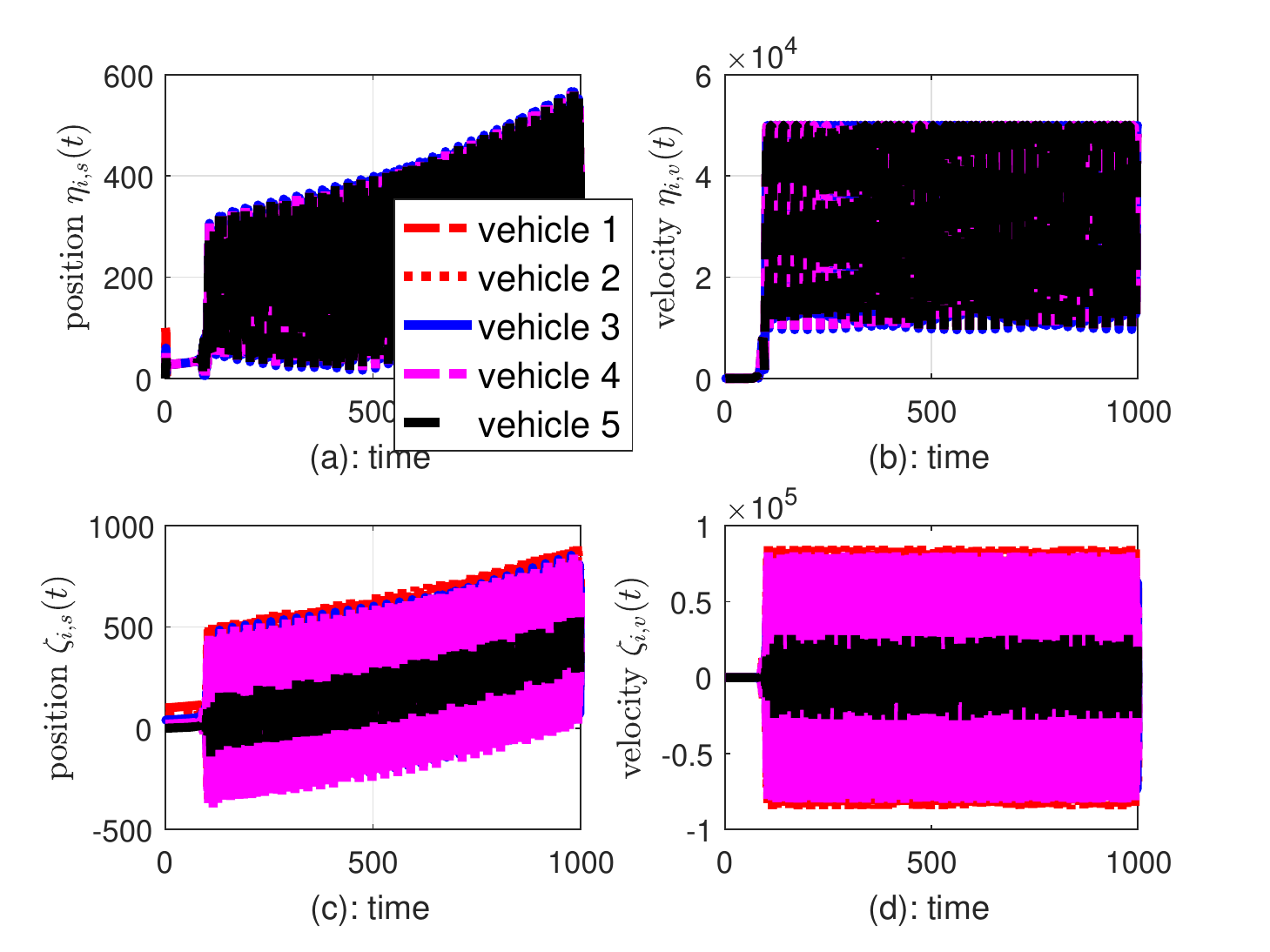}
		\caption{The conventional platooning of vehicles under attacks.}
		\label{fig:convent_estimation}
	\end{figure}
	\begin{figure}[t]
		\centering
		\includegraphics[width=1\linewidth]{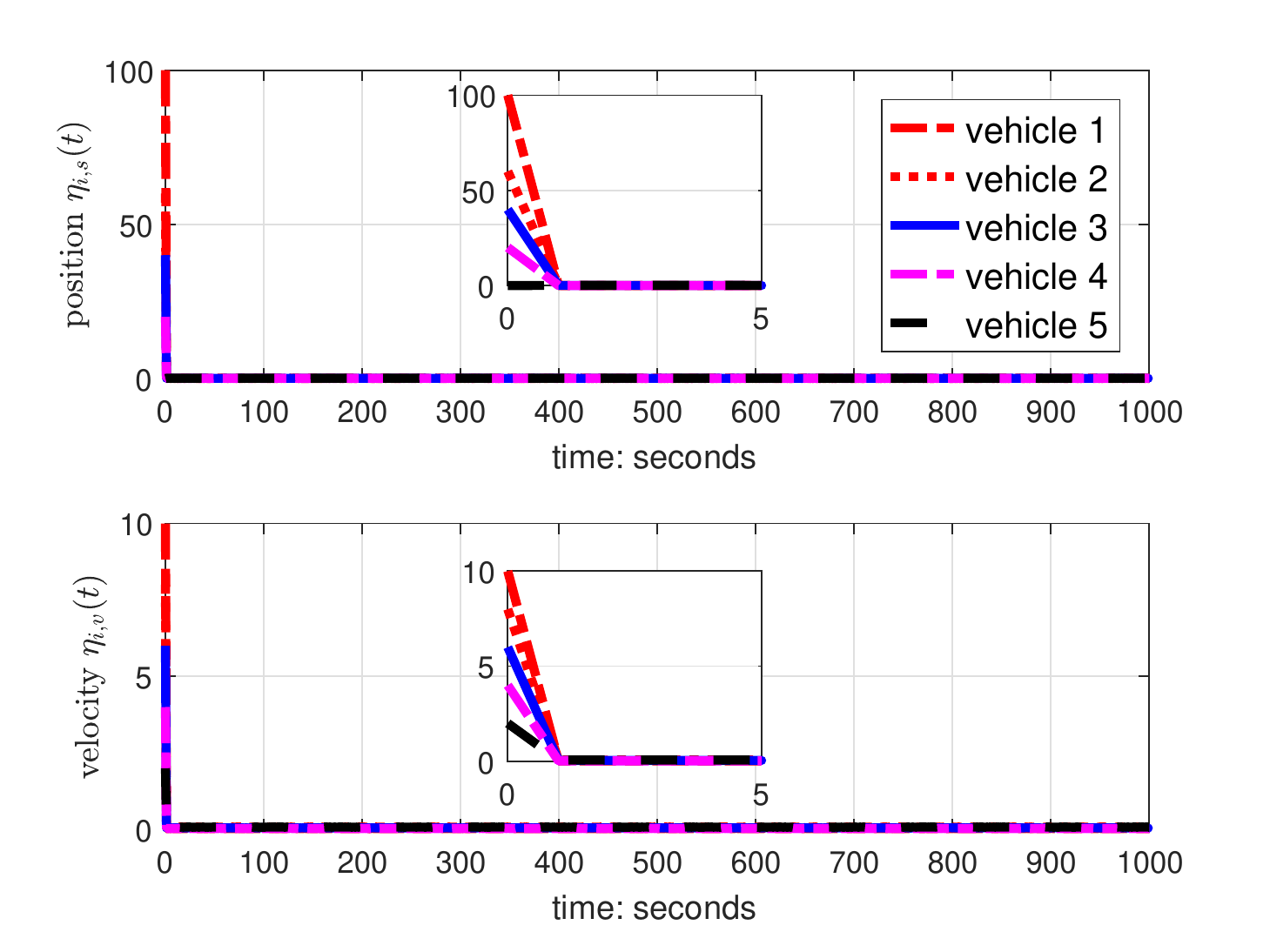}
		\caption{Estimation error of vehicles by Algorithm \ref{alg:A}.}
		\label{fig:estimation}
	\end{figure}
	\begin{figure}[t]
		\centering
		\includegraphics[width=1\linewidth]{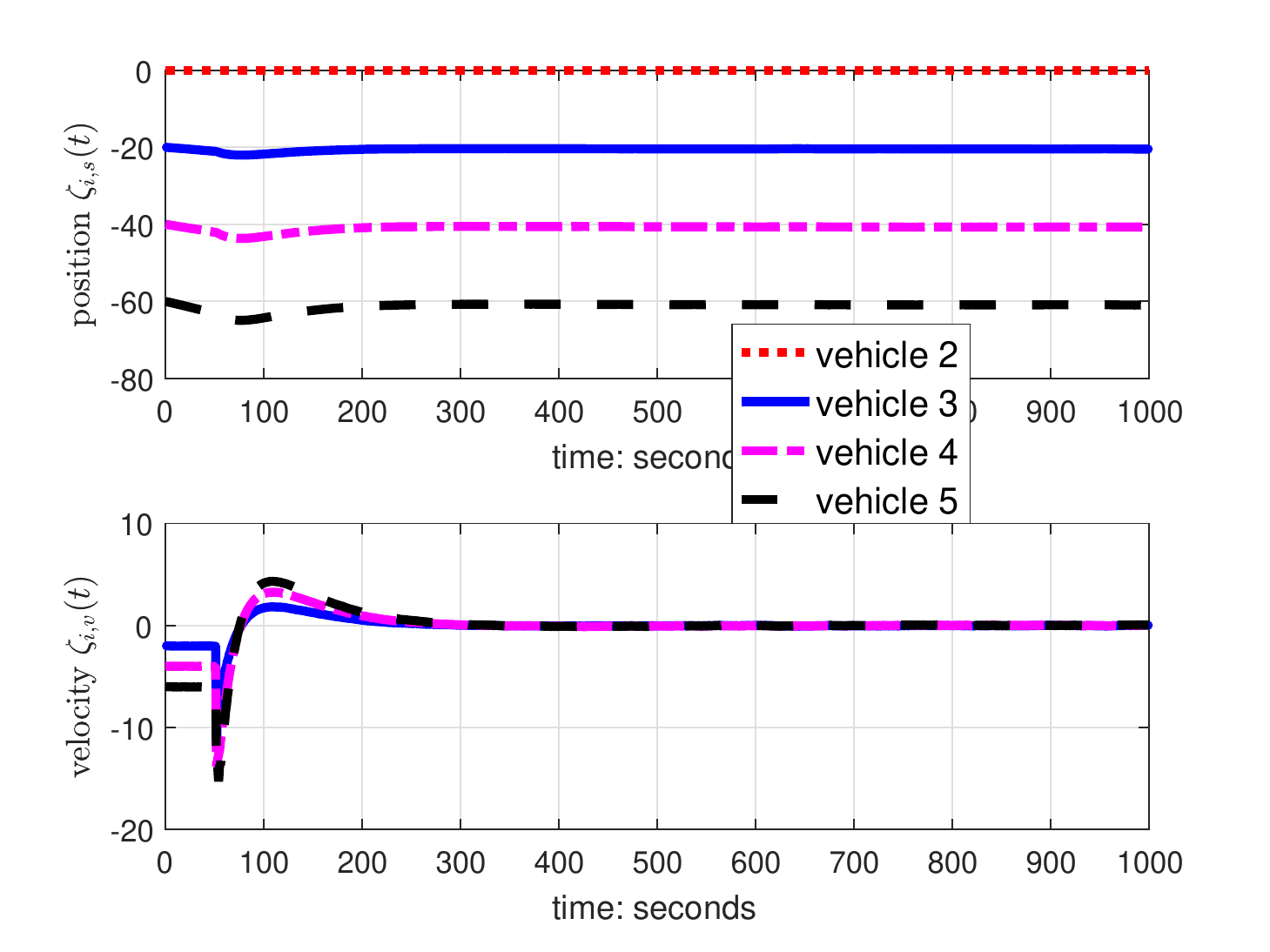}
		\caption{Platooning of vehicles by Algorithm \ref{alg:A}.}
		\label{fig:platoon}
	\end{figure}

	\section{Conclusion}\label{sec_conclusion}
	This paper studied how to design a secure algorithm such that a group of autonomous vehicles achieve platooning under the case that the GPS data of one  vehicle is under attack. We provided a distributed algorithm, which consists of a local state observer, two online attack detectors and a distributed controller. The conditions to ensure bounded state estimation error and vehicle platooning error  were studied.

	\appendix
	
	\subsection{Derivation of \eqref{eq_system5}}\label{app_derivation}
	For vehicle $i\in \{2,3,\dots,N-1\}$, by \eqref{eq_system2} and \eqref{eq_system3}, we have
	\begin{equation}\label{eq_system4}
	\begin{split}
	y_{i|i-1}(t)&=x_i(t)+ a_{i-1}(t)+d_{i|i-1}(t)\\
	y_{i|i+1}(t)&=x_i(t)+ a_{i+1}(t)+d_{i|i+1}(t)
	\end{split}
	\end{equation}
	where 
	\begin{align*}
	y_{i|i-1}(t)&=y_{i-1,i}(t)+y_{i-1,i-1}(t)\\
	y_{i|i+1}(t)&=y_{i+1,i+1}(t)-y_{i,i+1}(t)\\
	d_{i|i-1}(t)&=d_{i-1,i}(t)+d_{i-1,i-1}(t)\\
	d_{i|i+1}(t)&=d_{i+1,i+1}(t)-d_{i,i+1}(t).
	\end{align*}
	
	For vehicle $1$,
	\begin{equation}\label{eq_system_0}
	\begin{split}
	y_{1|2}(t)&=x_1(t)+ a_{2}(t)+d_{1|2}(t)\\
	y_{1|3}(t)&=x_1(t)+ a_{3}(t)+d_{1|3}(t)
	\end{split}
	\end{equation}
	where
	\begin{align*}
	y_{1|2}(t)&=y_{2,2}(t)-y_{1,2}(t)\\
	y_{1|3}(t)&=y_{3,3}(t)-y_{1,2}(t)-y_{2,3}(t)\\
	d_{1|2}(t)&=d_{2,2}(t)-d_{1,2}(t)\\
	d_{1|3}(t)&=d_{3,3}(t)-d_{1,2}(t)-d_{2,3}(t).
	\end{align*}

	For vehicle $N$,
	\begin{equation}\label{eq_system_1}
	\begin{split}
	y_{N|N-1}(t)&=x_N(t)+ a_{N-1}(t)+d_{N|N-1}(t)\\
	y_{N|N-2}(t)&=x_N(t)+ a_{N-2}(t)+d_{N|N-2}(t)
	\end{split}
	\end{equation}
	where 
	\begin{align*}
	y_{N|N-1}(t)&=y_{N-1,N}(t)+y_{N-1,N-1}(t)\\
	y_{N|N-2}(t)&=y_{N-1,N}(t)+y_{N-2,N-1}(t)+y_{N-2,N-2}(t)\\
	d_{N|N-1}(t)&=d_{N-1,N}(t)+d_{N-1,N-1}(t)\\
	d_{N|N-2}(t)&=d_{N-1,N}(t)+d_{N-2,N-1}(t)+d_{N-2,N-2}(t).
	\end{align*}
	
	From \eqref{eq_system4}--\eqref{eq_system_1}, \eqref{eq_system5} is obtained.

	\subsection{Proof of Lemma \ref{lem_detec2}}
	Let $e_i(t)=\hat x_i(t)-x_i(t)$, $i\in\{1,\dots,N\}$. 
	We use an inductive method to prove the two conclusions.		
	At the initial time, under Assumption \ref{ass_noise},  $\norm{e_i(0)}\leq q.$ 
	For an attack-free vehicle $i$, by equations \eqref{eq_local_state}--\eqref{eq_system2},
	\begin{align*}
	\norm{y_{i,i}(1)-A\hat x_{i}(0)}&=\norm{-Ae_i(0)+n_i(1)+d_{i,i}(1)}\\
	&\leq \norm{A}q+\epsilon+\mu.
	\end{align*}
	
	Assume at time $t-1$, 
	\begin{align*}
	\begin{split}
	\norm{e_i(t-1)}&\leq \rho(t-1)\\
	\norm{y_{i,i}(t)-A\hat x_{i}(t-1)}&\leq \norm{A}\rho(t-1)+\epsilon+\mu.
	\end{split}
	\end{align*}
	Denote 
	\begin{align*}
	\bar K_{i,\mathcal{A}^c}(t)&=:\diag\bigg\{k_{i,j_1}^{[1]}(t)\mathbb I_{j_1\in \mathcal{A}^c},k_{i,j_1}^{[2]}(t)\mathbb I_{j_1\in \mathcal{A}^c},\\
	&\qquad k_{i,j_2}^{[1]}(t)\mathbb I_{j_2\in \mathcal{A}^c},k_{i,j_2}^{[2]}(t)\mathbb I_{j_2\in \mathcal{A}^c},\\
	&\qquad k_{i,j_3}^{[1]}(t)\mathbb I_{j_3\in \mathcal{A}^c},k_{i,j_3}^{[2]}(t)\mathbb I_{j_3\in \mathcal{A}^c}\bigg\}\in\mathbb{R}^{6\times 6},\nonumber
	\end{align*}
	As we see, $\bar K_{i,\mathcal{A}^c}(t)$ is diagonal and its diagonal  elements are zero if the corresponding vehicles are under attack. 
	Then we define the complementary of $\bar K_{i,\mathcal{A}^c}(t)$ with respect to $\bar K_{i}(t)$:
	\begin{align*}
	\bar K_{i,\mathcal{A}}(t)=	\bar K_{i}(t)-\bar K_{i,\mathcal{A}^c}(t).
	\end{align*}
	
	Let $\bar K_{i}^{[j]}(t)$ be the $j$th diagonal element of $\bar K_{i}(t)$, $j=1,\dots,6$, and $\boldsymbol{d}_i^{[j]}(t)$ be the $j$th element of $\boldsymbol{d}_i(t)$. Denote
	\begin{align*}
	\hat  K_{i}(t)&=\begin{pmatrix}
	\sum\limits_{j=1,3,5}\bar K_{i}^{[j]}(t)&0\\
	0& \sum\limits_{j=2,4,6}\bar K_{i}^{[j]}(t)
	\end{pmatrix}\in\mathbb{R}^{2\times 2}\\
	W_i(t)&=\sum\limits_{j=1,3,5}\begin{pmatrix}
	\bar K_{i}^{[j]}(t)&0\\
	0&\bar K_{i}^{[j+1]}(t)
	\end{pmatrix}\begin{pmatrix}
	\boldsymbol{d}_i^{[j]}(t)\\
	\boldsymbol{d}_i^{[j+1]}(t)
	\end{pmatrix}\in\mathbb{R}^{2}.
	\end{align*}
	
	By Algorithm \ref{alg:A},  we have
	\begin{align*}
	e_i(t)=&Ae_i(t-1)+\frac{1}{2} C^T\bar K_i(t)(z_i(t)-C\bar x_i(t))\\
	=&(I_2-\frac{1}{2}\hat  K_{i}(t))Ae_i(t-1)+\frac{1}{2}\hat  K_{i}(t)n_i+\frac{1}{2}W_i(t)\\
	&+\frac{1}{2} C^T\bar K_{i,\mathcal{A}}(t)(z_i(t)-C\bar x_i(t)).
	\end{align*}
	
	Taking the norm operation of $e_i(t)$ yields
	\begin{align}\label{eq_int}
	\norm{e_i(t)}\leq& \norm{(I_2-\frac{1}{2}\hat  K_{i}(t))A}\norm{e_i(t-1)}\nonumber\\
	&+|\mathcal{A}_{i,c}|\frac{1}{2}(\epsilon+\mu)+|\mathcal{A}_{i}|\sqrt{2}\frac{1}{2}\beta\nonumber\\
	\leq &\norm{(I_2-\frac{1}{2}\hat  K_{i}(t))A}\norm{e_i(t-1)}+Q\nonumber\\
	\leq& (1-k(t))\norm{A}\rho(t-1)+Q\nonumber\\
	=&\rho(t).
	\end{align}

	For an attack-free vehicle $i$, by equations \eqref{eq_local_state}--\eqref{eq_system2},
	\begin{align*}
	\norm{y_{i,i}(t+1)-A\hat x_{i}(t)}&=\norm{-Ae_i(t)+n_i(t+1)+d_{i,i}(t+1)}\\
	&\leq \norm{A}\rho(t)+\epsilon+\mu.
	\end{align*}
	\subsection{Proof of Theorem \ref{thm_estimation}}
	From Lemma \ref{lem_detec2}, $\norm{\hat x_i(t)-x_i(t)}\leq \rho(t)$.
	The conclusion of this theorem holds if $\limsup \limits_{t\rightarrow \infty} \rho(t)\leq \alpha$, which will be proved in the following. 
	
	Let $k^*=\min\{1,\frac{\beta}{\norm{A}q+\epsilon+\mu}\}.$
	Next, we prove $\rho(t)< q$ for $t\geq 1$. From \eqref{k_seq} and $\rho(0)= q$,  $k(1)\geq k^*$.  By \eqref{sequ}, 
	$\rho(1)=(1-k(1))\norm{A}\rho(0)+Q
	\leq(1-k^*)\norm{A}q+Q
	<q.$
	Suppose at time $t$, $\rho(t)< q$. 
	At time $t+1$, from \eqref{k_seq} and  $k(t+1)>k^*$. By \eqref{sequ},  $\rho(t+1)<(1-k^*)\norm{A}q+Q<q.$
	Thus, $\rho(t)< q$ for $t\geq 1.$

	%
	%
	%
	%
	
	Note that $0<\norm{(1-k^*)A}<1$, if \eqref{condition0} is satisfied and $k^*\neq 1$. As a result, for $i\in\{1,\dots,N\},$ we have
	\begin{align}\label{alpha_2}
	&\limsup\limits_{t\rightarrow \infty}\rho(t)\nonumber\\
	\leq& \alpha:=Q\left(\mathbb I_{k^*=1}+\frac{1}{1-\norm{(1-k^*)A}}\mathbb I_{k^*\neq 1}\right),
	\end{align}
	where $Q=\frac{3}{2}(\epsilon+\mu)+\frac{\sqrt{2}}{2}\beta.$
	
	\subsection{Proof of Theorem \ref{thm_control}}
	Suppose $x_i^*(t)=[s_i^*(t),v_i^*(t)]^T$ is the desired state of vehicle $i$. 
	Denote the tracking error of vehicle $i$,  $2\leq i\leq N$,  by $\tilde x_i(t)=[\tilde s_i(t),\tilde v_i(t)]^T$ which satisfies the following equation
	\begin{align}\label{system_control}
	\tilde x_i(t+1)&=A\tilde x_i(t)+[0,T	\tilde  u_i(t)]^T+\delta_i(t)\\
	\delta_i(t)&=[0,T\hat u_i(t)]^T+n_{i}(t)
	\end{align}
	where
	\begin{align}
	\begin{split}
	\tilde u_{i}(t)=&\sum_{j\in\mathcal{N}_i}\big(g_s(	\tilde s_{j}(t)-	\tilde s_{i}(t))\\
	&+g_v( 	\tilde v_{j}(t)-	\tilde v_{i}(t))\big), 2\leq i\leq N, 1\leq j\leq N,\\
	\hat u_i(t)=&\sum_{j\in\mathcal{N}_i}\big(g_s(	(\bar s_{j}(t)-s_{j}(t))-	(\hat s_{i}(t)- s_{i}(t)))\\
	&+g_v( 	(\bar v_{j}(t)-v_{j}(t))-	(\hat v_{i}(t)- v_{i}(t))\big)
	\end{split}
	\end{align}
	Let $\tilde X(t)=[\tilde x_1(t)^T,\tilde x_2(t)^T,\dots,\tilde x_N(t)^T]$ and $\delta(t)=[\delta_1(t)^T,\delta_2(t)^T,\dots,\delta_N(t)^T]$, then we have
	\begin{align}\label{eq_Error}
	\tilde X(t+1)=P\tilde X(t)+\delta(t),
	\end{align}
	where $P=I_{N-1}\otimes A-\mathcal{L}_g\otimes F$, $F=\begin{pmatrix}
	0&0\\
	T g_s&	T g_v
	\end{pmatrix}$ and $\mathcal{L}_g$ is the ground Laplacian matrix with respect to the nodes $\{2,3,\dots,N\}$ obtained by removing the first row and first column of the Laplacian matrix $\mathcal{L}$. 
	
	By Theorem \ref{thm_estimation}, $\sup_{t\geq 0}\norm{\delta(t)}<\infty$. Based on the BIBO stability principle, the stability of $\tilde X(t)$ in \eqref{eq_Error} depends on $P$. According to \cite{hao2010effect}, the spectrum of $P$ is 
	\begin{align}
	\sigma(P)&=\bigcup_{\lambda_{l}\in \sigma(\mathcal{L}_g)}\{A-\lambda_{l}F\}\nonumber\\
	&=\bigcup_{\lambda_{l}\in \sigma(\mathcal{L}_g)}\Bigg\{\sigma\begin{pmatrix}
	1&T\\
	-\lambda_l Tg_s& 1-\lambda_lTg_v
	\end{pmatrix}\Bigg\}
	\end{align}
	where $\sigma(\cdot)$ is the set of distinct  eigenvalues. 
	
	Let $Q=\begin{pmatrix}
	1&T\\
	-\lambda_l Tg_s& 1-\lambda_lTg_v
	\end{pmatrix}$. The eigenvalue of $Q$, denoted by $s$, is the root of the following equation
	\begin{align}\label{charastic}
	s^2+(\lambda_lTg_v-2)s+\lambda_l T^2g_s-\lambda_lTg_v+1=0.
	\end{align}
	By  \cite{hao2010effect}, all eigenvalues of $\mathcal{L}_g$ are positive, i.e., $\lambda_{l}>0$.
	Under the condition in \eqref{eq_condition}, we have
	\begin{align*}
	&(\lambda_lTg_v-2)^2-4(\lambda_l T^2g_s-\lambda_lTg_v+1)\\
	=&\lambda_l^2T^2g_v^2-4\lambda_l T^2g_s\\
	= &\lambda_lT^2(\lambda_lg_v^2-4 g_s)\\
	\leq&\lambda_lT^2(\lambda_{\max}(\mathcal{L}_g)g_v^2-4 g_s)\\
	<&0.
	\end{align*}
	which results in that \eqref{charastic} has two conjugate complex roots. The two roots  share the same modulus in the following
	$|s|=\lambda_l T^2g_s-\lambda_lTg_v+1.$
	Due to \eqref{eq_condition} and $\lambda_l >0$, $|s|<1.$ The  eigenvalues of $Q$ fall into the open unit disk for each $\lambda_l.$ 
	Therefore, $P$ is Schur stable.
	
	Under Theorem \ref{thm_estimation}, it can be proved that there is a constant scalar $\alpha_3>0$, such that
	$	\limsup\limits_{t\rightarrow \infty}\norm{\delta(t)}\leq \alpha_3.$
	By Lemma \ref{lem_stability} and \eqref{eq_Error},  there is a constant scalar $\alpha_4>0$, such that
	$\limsup\limits_{t\rightarrow \infty}\norm{\tilde X(t)}\leq \alpha_4,$
	which ensures the conclusion of this theorem.

	\subsection{Useful lemma}
	\begin{lemma}\label{lem_stability}
		Consider the  linear dynamical system 
		$	x(t+1)=Fx(t)+G(t),$
		where $F\in\mathbb{R}^{n\times n}$ is a Schur stable matrix. Then there is a positive definite matrix $P,$ satisfying $F^TPF-P=-I_n$, such that
		\begin{align}\label{eq_x}
		\norm{x(t)}^2\leq& \frac{\lambda^t\lambda_{\max}(P)\norm{x(0)}^2}{\lambda_{\min}(P)}\nonumber\\
		&+\frac{\beta}{\lambda_{\min}(P)}\sum_{l=0}^{t-1}\lambda^{t-1-l}\norm{G(l)}^2.
		\end{align}
		where $\lambda=1-\frac{1}{2\lambda_{\max}(P)}\in(0,1)$ and 	$\beta=\norm{P}+2\norm{PF}^2.$ 
		Furthermore,
		\begin{enumerate}
			\item if $\sup_{t\geq 0}\norm{G(t)}\leq \alpha_1$, then
			$$\limsup\limits_{t\rightarrow \infty}\norm{x(t)}^2\leq \frac{\beta\alpha_1^2}{\lambda_{\min}(P)(1-\lambda)};$$
			\item if $\limsup\limits_{t\rightarrow \infty}\norm{G(t)}\leq \alpha_2$, then 
			$$		\limsup\limits_{t\rightarrow \infty}\norm{x(t)}^2\leq \frac{\beta\alpha_2^2}{\lambda_{\min}(P)(1-\lambda)}.$$
		\end{enumerate}
	\end{lemma}
	\begin{proof}
		First, it can be easily verified that $P=\sum_{i=0}^{\infty}(F^i)^TF^i$ is the solution of $F^TPF-P=-I_n$ with $\lambda_{\min}(P)\geq 1$.
		Let $V(x(t))=x(t)^TPx(t)$. Then we consider 
		\begin{align}
		&V(x(t+1))-V(x(t))\nonumber\\
		=&G(t)^TPG(t)+2x(t)^TF^TPG(t)-x(t)^Tx(t)\nonumber\\
		\overset{(a)}{\leq} & \norm{P}\norm{G(t)}^2-\frac{1}{2}x(t)^Tx(t)+2G(t)^TPFF^TPG(t)\nonumber\\
		\leq& \beta\norm{G(t)}^2-\frac{1}{2\lambda_{\max}(P)}V(x(t)),\label{eq_derive}
		\end{align}
		where $(a)$ is obtained by noting that $(\frac{1}{\sqrt{2}}x(t)-\sqrt{2}F^TPG(t))^T(\frac{1}{\sqrt{2}}x(t)-\sqrt{2}F^TPG(t))\geq 0.$
		By \eqref{eq_derive}, we then have
		\begin{align}\label{eq_v}
		V(x(t+1))\leq \lambda V(x(t))+\beta \norm{G(t)}^2.
		\end{align}
		By applying \eqref{eq_v} for $t$ times, we have
		\begin{align}\label{eq_v2}
		V(x(t))\leq \lambda ^tV(x(0))+\beta\sum_{l=0}^{t-1}\lambda^{t-1-l}\norm{G(l)}^2.
		\end{align}
		The equation \eqref{eq_x} is straightforward to be obtained by \eqref{eq_v2}. 
		
		The conclusion in 1) is obtained by \eqref{eq_x}  and $\lambda\in(0,1)$. Next, we prove 2).
		To prove 2), we just need to prove that for $\forall \epsilon_0>0$, there is a constant $T_0>0$, such that for any $t\geq T_0$, $\norm{x(t)}^2\leq \frac{\beta\alpha_2^2}{\lambda_{\min}(P)(1-\lambda)}+\epsilon_0.$ 
		
		Due to $\limsup\limits_{t\rightarrow \infty}\norm{G(t)}\leq \alpha_2$, $\forall \epsilon_1>0$, there is a constant $T_1>0$, such that for any $t\geq T_1$, $\norm{G(t)}\leq \alpha_2+\epsilon_1$.
		Then we rewrite \eqref{eq_x} in the following
		\begin{align}\label{eq_x2}
		\norm{x(t)}^2\leq& \frac{\lambda^t\lambda_{\max}(P)\norm{x(0)}^2}{\lambda_{\min}(P)}\nonumber\\
		&+\frac{\beta}{\lambda_{\min}(P)}\sum_{l=0}^{T_1-1}\lambda^{t-1-l}\norm{G(l)}^2\nonumber\\
		&+\frac{\beta}{\lambda_{\min}(P)}\sum_{l=T_1}^{t-1}\lambda^{t-1-l}\norm{G(l)}^2.
		\end{align}
		We analyze the three terms on the right-hand side of \eqref{eq_x2}, respectively.
		First, due to $\lambda\in(0,1)$, there is a time $T_2>0$, such that $t\geq T_2$,
		\begin{align}\label{eq_first}
		\frac{\lambda^t\lambda_{\max}(P)\norm{x(0)}^2}{\lambda_{\min}(P)}<\frac{\epsilon_0}{3}.
		\end{align} 	
		Second, due to $\sup_{t\geq 0}\norm{G(l)}^2<\infty,$ there is a time $T_3>0$, such that $t\geq T_3$,
		\begin{align}\label{eq_second}
		\frac{\beta}{\lambda_{\min}(P)}\sum_{l=0}^{T_1-1}\lambda^{t-1-l}\norm{G(l)}^2<\frac{\epsilon_0}{3}.
		\end{align}
		Finally, recall that $ \forall \epsilon_1>0$, $\norm{G(t)}\leq \alpha_2+\epsilon_1$ for any $t\geq T_1+1$. By taking $\epsilon_1<\frac{\epsilon_0\lambda_{\min}(P)(1-\lambda)}{3\beta}$, we have
		\begin{align}\label{eq_third}
		&\frac{\beta}{\lambda_{\min}(P)}\sum_{l=T_1}^{t-1}\lambda^{t-1-l}\norm{G(l)}^2\nonumber\\
		\leq &\frac{\beta(\alpha_2+\epsilon_1)}{\lambda_{\min}(P)}\frac{(1-\lambda^{t-T_1})}{1-\lambda}\nonumber\\
		\leq& \frac{\beta\alpha_2}{\lambda_{\min}(P)(1-\lambda)}+\frac{\beta\epsilon_1}{\lambda_{\min}(P)(1-\lambda)}\nonumber\\
		<&\frac{\beta\alpha_2}{\lambda_{\min}(P)(1-\lambda)}+\frac{\epsilon_0}{3}.
		\end{align}
		Let $T_0=\max\{T_1+1,T_2,T_3\}$. Then given $\forall \epsilon_0>0$,  for any $t\geq T_0$, by \eqref{eq_x2}--\eqref{eq_third}, $\norm{x(t)}^2\leq \frac{\beta\alpha_2^2}{\lambda_{\min}(P)(1-\lambda)}+\epsilon_0.$ 	
	\end{proof}

	\bibliography{All_references}
	\bibliographystyle{ieeetr}
	
\end{document}